\documentclass[11pt]{article}
\usepackage{hyperref}
\usepackage{times}  
\usepackage{mathpazo}
\usepackage{amssymb,amsmath,amsthm}
\usepackage{epsfig}

\newtheorem{theorem}{Theorem}[section]
\newtheorem{proposition}[theorem]{Proposition}
\newtheorem{definition}[theorem]{Definition}

\newtheorem{claim}[theorem]{Claim}
\newtheorem{lemma}[theorem]{Lemma}
\newtheorem{conjecture}[theorem]{Conjecture}

\newtheorem{corollary}[theorem]{Corollary}

\newtheorem{remark}[theorem]{Remark}
\newtheorem{question}[theorem]{Question}

\newcommand{\qedsymb}{\hfill{\rule{2mm}{2mm}}}
\renewenvironment{proof}[1][]{\begin{trivlist}
\item[\hspace{\labelsep}{\bf\noindent Proof#1:\/}] }{\qedsymb\end{trivlist}}

\def\calA{{\cal A}}

\def\calF{{\cal F}}
\def\calL{{\cal L}}

\def\R{\mathbb{R}}

\def\N{\mathbb{N}}

\newcommand{\NP}{\mathsf{NP}}

\newcommand{\DTIME}{\mathsf{DTIME}}

\newcommand{\od}{\overline{\xi}}

\newcommand{\eps}{\epsilon}
\renewcommand{\epsilon}{\varepsilon}

\newcommand{\rank}{\mathop{\mathrm{rank}}}

\newcommand{\poly}{\mathop{\mathrm{poly}}}

\newcommand{\Fset}{\mathbb{F}}         

\newcommand{\vchrom}{{\chi_\mathrm{v}}}
\newcommand{\svchrom}{{\chi^{(\mathrm{s})}_\mathrm{v}}}

\begin{document}

\title{{\bf Approximating the Orthogonality Dimension of Graphs and Hypergraphs}}

\author{
Ishay Haviv\thanks{School of Computer Science, The Academic College of Tel Aviv-Yaffo, Tel Aviv 61083, Israel.
}
}

\date{}

\maketitle

\begin{abstract}
A $t$-dimensional orthogonal representation of a hypergraph is an assignment of nonzero vectors in $\R^t$ to its vertices, such that every hyperedge contains two vertices whose vectors are orthogonal. The orthogonality dimension of a hypergraph $H$, denoted by $\od(H)$, is the smallest integer $t$ for which there exists a $t$-dimensional orthogonal representation of $H$.
In this paper we study computational aspects of the orthogonality dimension of graphs and hypergraphs.
We prove that for every $k \geq 4$, it is $\NP$-hard (resp. quasi-$\NP$-hard) to distinguish $n$-vertex $k$-uniform hypergraphs $H$ with $\od(H) \leq 2$ from those satisfying $\od(H) \geq \Omega(\log^\delta n)$ for some constant $\delta>0$ (resp. $\od(H) \geq \Omega(\log^{1-o(1)} n)$). For graphs, we relate the $\NP$-hardness of approximating the orthogonality dimension to a variant of a long-standing conjecture of Stahl.
We also consider the algorithmic problem in which given a graph $G$ with $\od(G) \leq 3$ the goal is to find an orthogonal representation of $G$ of as low dimension as possible, and provide a polynomial time approximation algorithm based on semidefinite programming.
\end{abstract}

\section{Introduction}

A $t$-dimensional {\em orthogonal representation} of a hypergraph $H=(V,E)$ is an assignment of a nonzero real vector $u_v \in \R^t$ to every vertex $v \in V$, such that every hyperedge $e \in E$ contains two vertices $v,v' \in e$ whose vectors $u_v$ and $u_{v'}$ are orthogonal. The {\em orthogonality dimension} of $H$, denoted by $\od(H)$, is the smallest integer $t$ for which there exists a $t$-dimensional orthogonal representation of $H$.\footnote{Orthogonal representations of {\em graphs} are sometimes defined in the literature as orthogonal representations of the complement, namely, the definition requires vectors associated with {\em non-adjacent} vertices to be orthogonal. In this paper we have decided to use the other definition because it is extended more naturally to hypergraphs. For a graph $G$, one can view the notation $\od(G)$ as standing for $\xi(\overline{G})$.}
The notion of orthogonal representations was introduced for graphs by Lov\'asz~\cite{Lovasz79} in the study of the Shannon capacity and was later involved in a geometric characterization of connectivity properties of graphs by Lov\'asz, Saks, and Schrijver~\cite{LovaszSS89}.
The orthogonality dimension over the complex field was used by de Wolf~\cite{deWolfThesis} in a characterization of the quantum one-round communication complexity of promise equality problems and by Cameron et al.~\cite{CameronMNSW07} in the study of the quantum chromatic number of graphs (see also~\cite{ScarpaS12,BrietBLPS15,BrietZ17}).
An extension of orthogonal representations, called orthogonal bi-representations, was introduced by Haemers~\cite{Haemers81} and has found several further applications to information theory and to theoretical computer science.

Orthogonal representations can be viewed as a generalization of hypergraph vertex colorings, one of the most fundamental and extensively studied topics in graph theory.
Recall that a hypergraph $H$ is said to be $c$-colorable if one can assign one of $c$ colors to every vertex of $H$ such that no hyperedge is monochromatic. The chromatic number of $H$, denoted by $\chi(H)$, is the smallest integer $c$ for which $H$ is $c$-colorable.
Obviously, every $c$-coloring of $H$ induces a $c$-dimensional orthogonal representation of $H$ by assigning the $i$th unit vector $e_i \in \R^c$ to every vertex colored by the $i$th color.
On the other hand, given a $t$-dimensional orthogonal representation $(u_v)_{v \in V}$ of $H$ one can assign to every vertex $v$ the vector in $\{-1,0,+1\}^t$ that consists of the signs of the entries of $u_v$, and since nonzero orthogonal vectors have distinct sign vectors it follows that $H$ is $3^t$-colorable. We conclude that every hypergraph $H$ satisfies
\begin{eqnarray}\label{eq:chi_vs_od}
\log_3 \chi(H) \leq  \od (H) \leq \chi (H).
\end{eqnarray}
The upper bound in~\eqref{eq:chi_vs_od} can clearly be tight (take, e.g., a complete graph), and it turns out that there exist graphs whose orthogonality dimension is exponentially smaller than their chromatic number (see Proposition~\ref{prop:od_vs_chi}).

The current work studies the problem of approximating the orthogonality dimension of graphs and hypergraphs.
This research direction was already suggested in the late eighties by Lov\'asz et al.~\cite{LovaszSS89}, who remarked that computing the orthogonality dimension of graphs seems to be a difficult task (see also~\cite{LovaszLN}). Nevertheless, the only hardness result we are aware of for this graph parameter is the one of Peeters~\cite{Peeters96}, who proved that for every $t \geq 3$ it is $\NP$-hard to decide whether an input graph $G$ satisfies $\od(G) \leq t$ (see also~\cite{BrietBLPS15}).
Before stating our hardness and algorithmic results, we overview related previous work on graph and hypergraph coloring.

\subsection{Graph and Hypergraph Coloring}\label{sec:coloring}

It is well known that the problem of deciding whether an input graph is $c$-colorable can be easily solved in polynomial time for $c \in \{1,2\}$ and is $\NP$-hard for every $c \geq 3$~\cite{Karp72}.

In 1976, Garey and Johnson~\cite{GareyJ76a} have discovered an interesting connection between hardness of graph coloring and the multichromatic numbers of Kneser graphs.
For integers $d \geq 2s$, the {\em Kneser graph} $K(d,s)$ is the graph whose vertices are all the $s$-subsets of $[d]$, where two sets are adjacent if they are disjoint.
A {\em $k$-tuple coloring} $f$ of a graph $G=(V,E)$ is an assignment $f(v)$ of a set of $k$ colors to every vertex $v \in V$ such that $f(v) \cap f(v') = \emptyset$ whenever $v$ and $v'$ are adjacent in $G$. The {\em $k$th multichromatic number} of $G$, denoted by $\chi_k(G)$, is the smallest integer $c$ for which $G$ has a $k$-tuple coloring with $c$ colors.
Equivalently, $\chi_k(G)$ is the smallest integer $c$ for which there exists a homomorphism from $G$ to the Kneser graph $K(c,k)$. Note that $\chi_1(G)$ is simply the standard chromatic number $\chi(G)$.
In the seventies, Stahl~\cite{Stahl76} has made the following conjecture regarding the multichromatic numbers of Kneser graphs (see also~\cite{FranklF86a}).

\begin{conjecture}[\cite{Stahl76}]\label{conj:Stahl}
For all integers $k$ and $d \geq 2s$,
\[\chi_k(K(d,s)) = \Big \lceil \frac{k}{s} \Big \rceil \cdot (d-2s)+2k.\]
\end{conjecture}
\noindent
More than forty years later, Conjecture~\ref{conj:Stahl} it is still widely open.
While the right-hand side in the conjecture is known to form an upper bound on $\chi_k(K(d,s))$ for all values of $k$, $d$ and $s$, the conjecture was confirmed only for a few special cases.
For $k=1$, the conjecture was proved by Lov\'asz~\cite{LovaszKneser} in a breakthrough application of algebraic topology confirming a conjecture by Kneser~\cite{Kneser55}.
Stahl~\cite{Stahl76} proved that the conjecture holds whenever $1 \leq k \leq s$, $d=2s+1$, or $k$ is divisible by $s$.
Garey and Johnson~\cite{GareyJ76a} proved the case of $s=3$ and $k=4$, namely, that $\chi_4(K(d,3))=2d-4$, and used it in the analysis of a reduction from $3$-colorability to prove that for every $c \geq 6$, it is $\NP$-hard to distinguish between graphs $G$ that satisfy $\chi(G) \leq c$ and those that satisfy $\chi(G) \geq 2c-4$.

In 1993, Khanna, Linial, and Safra~\cite{KhannaLS00} proved that it is $\NP$-hard to decide whether an input graph $G$ satisfies $\chi(G) \leq 3$ or $\chi(G) \geq 5$ (see also~\cite{GuruswamiK04}). As observed in~\cite{BrakensiekG16}, combining this result with the proof technique of~\cite{GareyJ76a} and the case of $s=3$ and $k=5$ in Conjecture~\ref{conj:Stahl} proved by Stahl~\cite{Stahl98} (who confirmed there the conjecture for $s \leq 3$ and all integers $k$ and $d$), it follows that for every $c \geq 6$ it is $\NP$-hard to distinguish between the cases $\chi(G) \leq c$ and $\chi(G) \geq 2c-2$.
Brakensiek and Guruswami~\cite{BrakensiekG16} improved this result using different techniques and proved the $\NP$-hardness of deciding whether a given graph $G$ satisfies $\chi(G) \leq c$ or $\chi(G) \geq 2c-1$ for all $c \geq 3$. In a recent work of Bul{\'{\i}}n, Krokhin, and Opr{\v s}al~\cite{BKO19}, the latter condition was further improved to $\chi(G) \geq 2c$.
We note that Dinur, Mossel, and Regev~\cite{DinurMR06} proved that assuming a certain variant of the unique games conjecture, deciding whether a given graph $G$ satisfies $\chi(G) \leq 3$ or $\chi(G) \geq c$ is $\NP$-hard for every $c \geq 4$.

We next consider, for any constant $k \geq 3$, the problem of deciding whether an input $k$-uniform hypergraph (i.e., a hypergraph each of its hyperedges contains exactly $k$ vertices) is $c$-colorable.
While the problem can clearly be solved in polynomial time for $c=1$, it was shown to be $\NP$-hard for $c=2$ and $k = 3$ by Lov\'asz~\cite{Lovasz73}, motivating the study of the following gap problem:
Given an $n$-vertex $k$-uniform hypergraph $H$, decide whether $\chi(H) \leq 2$ or $\chi(H) \geq c$.
Guruswami, H{\aa}stad, and Sudan~\cite{GuruswamiHS02} proved that the problem is $\NP$-hard for $k \geq 4$ and every constant $c \geq 3$. By combining their proof with the later PCP theorem of Moshkovitz and Raz~\cite{MoshkovitzR10}, this result also follows for the super-constant $c= \Omega(\frac{\log \log n}{\log \log \log n})$.
For $k=3$, the $\NP$-hardness was proved for every constant $c \geq 3$ by Dinur, Regev, and Smyth~\cite{DinurRS05}.
Their proof approach was extended in a recent work of Bhangale~\cite{Bhangale18}, who obtained $\NP$-hardness for every $k \geq 4$ with the improved super-constant $c = \Omega(\log^\delta n)$ where $\delta >0$ is some constant.
Under the complexity assumption $\NP \nsubseteq \DTIME(2^{\poly (\log n)})$, several stronger hardness results are known.
This includes the case of $k=3$ and $c = \Omega(\sqrt[3]{ \log \log n})$ proved in~\cite{DinurRS05} and the case of $k \geq 4$ and $c = \Omega(\frac{\log n}{\log \log n})$ proved in~\cite{Bhangale18}.
For additional related results see, e.g.,~\cite[Table~1]{Bhangale18} and the references therein.

On the algorithmic side, significant efforts have been made in the literature to obtain polynomial time algorithms for coloring $n$-vertex $3$-colorable graphs using as few colors as possible. This line of research was initiated by a simple algorithm of Wigderson~\cite{Wigderson83} that used $O(\sqrt{n})$ colors. In a series of increasingly sophisticated combinatorial algorithms, Blum~\cite{Blum94} improved the number of colors to $\widetilde{O}(n^{3/8})$. Then, Karger, Motwani, and Sudan~\cite{KargerMS98} introduced an algorithm based on a semidefinite relaxation and improved the number of colors to $\widetilde{O}(n^{1/4})$. Combining the combinatorial approach of~\cite{Blum94} and the semidefinite relaxation of~\cite{KargerMS98}, Blum and Karger~\cite{Blum94,BlumK97} improved it to $\widetilde{O}(n^{3/14})$, which was later improved by Arora, Chlamtac, and Charikar~\cite{AroraCC06} and by Chlamtac~\cite{Chlamtac07} to $\widetilde{O}(n^{0.2111})$ and $\widetilde{O}(n^{0.2072})$ respectively. The combinatorial component of these algorithms was recently improved by Kawarabayashi and Thorup~\cite{KT17}, who reduced the number of colors to $\widetilde{O}(n^{0.19996})$.
Halperin et al.~\cite{HalperinNZ02} have obtained analogue results for coloring $n$-vertex $c$-colorable graphs for all constants $c \geq 4$, e.g., for $c=4$ there exists an efficient algorithm that uses $\widetilde{O}(n^{7/19})$ colors.

For hypergraphs, there exists a simple combinatorial algorithm that given an $n$-vertex $k$-uniform $2$-colorable hypergraph finds in polynomial time a coloring with $\widetilde{O}(n^{1-1/k})$ colors, as was shown independently by Alon et al.~\cite{AlonKMH96} and by Chen and Frieze~\cite{ChenF96}. For $k=3$, this algorithm was combined in~\cite{AlonKMH96,ChenF96} with the semidefinite programming approach of~\cite{KargerMS98} to obtain a better bound of $\widetilde{O}(n^{2/9})$, which was later improved to $\widetilde{O}(n^{1/5})$ by Krivelevich, Nathaniel, and Sudakov~\cite{KrivelevichNS01}.
We note, however, that Alon et al.~\cite{AlonKMH96} have provided evidence that the powerful semidefinite approach cannot be applied to coloring $k$-uniform hypergraphs for $k \geq 4$.

\subsection{Our Contribution}

The present paper offers hardness and algorithmic results on the orthogonality dimension of graphs and hypergraphs.
We first mention that known hardness results on the chromatic number can be used to derive hardness results for the orthogonality dimension. Indeed, the inequalities given in~\eqref{eq:chi_vs_od} yield that for every integers $t_1$ and $t_2$, $\NP$-hardness of deciding whether an input $k$-uniform hypergraph $H$ satisfies $\chi(H) \leq t_1$ or $\chi(H) \geq t_2$ immediately implies the $\NP$-hardness of deciding whether it satisfies $\od(H) \leq t_1$ or $\od(H) \geq \log_3 t_2$. In particular, the hardness results of~\cite{DinurRS05,GuruswamiK04} imply that for all constants $k \geq 3$ and $t \geq 3$, it is $\NP$-hard to decide whether an input $k$-uniform hypergraph $H$ satisfies $\od(H) \leq 2$ or $\od(H) \geq t$.
For $k=2$, such a result follows from~\cite{DinurMR06} under a variant of the unique games conjecture.
However, for super-constant hardness gaps this implication leads to an exponential loss. In particular, it follows from~\cite{Bhangale18} that for every $k \geq 4$ it is $\NP$-hard to distinguish $n$-vertex $k$-uniform hypergraphs $H$ with $\od(H) \leq 2$ from those satisfying $\od(H) \geq \Omega(\log \log n)$.
We prove that this exponential loss can be avoided.

\begin{theorem}\label{thm:intro-NP-hard-gen}
Let $k \geq 4$ be a fixed integer.
\begin{enumerate}
  \item There exists a constant $\delta>0$ for which it is $\NP$-hard to decide whether an input $n$-vertex $k$-uniform hypergraph $H$ satisfies $\od(H) \leq 2$ or $\od(H) \geq \log ^\delta n$.
  \item Assuming $\NP \nsubseteq \DTIME(n^{O(\log \log n)})$, for every constant $c>0$ there is no polynomial time algorithm that decides whether an input $n$-vertex $k$-uniform hypergraph $H$ satisfies $\od(H) \leq 2$ or $\od(H) \geq c \cdot \frac{\log n}{\log \log n}$.
\end{enumerate}
\end{theorem}

We next consider the hardness of approximating the orthogonality dimension of {\em graphs}.
Our result involves a generalization of orthogonal representations of graphs defined as follows.
A $t$-dimensional {\em orthogonal $k$-subspace representation} of a graph $G=(V,E)$ is an assignment of a subspace $U_v \subseteq \R^t$ with $\dim (U_v)=k$ to every vertex $v \in V$, such that the subspaces $U_v$ and $U_{v'}$ are orthogonal whenever $v$ and $v'$ are adjacent in $G$. For a graph $G$, let $\od_k(G)$ denote the smallest integer $t$ for which there exists a $t$-dimensional orthogonal $k$-subspace representation of $G$. Note that $\od_1(G) = \od(G)$ for every graph $G$.
We prove the following result.

\begin{theorem}\label{thm:hard_F}
For every graph $F$, it is $\NP$-hard to decide whether an input graph $G$ satisfies $\od(G) \leq \od_3(F)$ or $\od(G) \geq \od_4(F)$.
\end{theorem}

With Theorem~\ref{thm:hard_F} in hand, it is of interest to find graphs $F$ for which $\od_4(F)$ is large compared to $\od_3(F)$.
We consider here, in light of Conjecture~\ref{conj:Stahl}, the behavior of the $\od_k$ parameters on the Kneser graphs $K(d,s)$.
For $k=1$, it was recently shown that the orthogonality dimension and the chromatic number coincide on Kneser graphs~\cite{Haviv18topo}.
We further observe, as an application of a result of Bukh and Cox~\cite{BukhC18}, that the values of $\chi_k$ and $\od_k$ coincide on the Kneser graphs $K(d,s)$ for every $k$ divisible by $s$ (that is, for all integers $\ell \geq 1$ and $d \geq 2s$,~ $\od_{\ell \cdot s}(K(d,s)) = \chi_{\ell \cdot s}(K(d,s)) = \ell \cdot d$, and in particular $\od_3(K(d,3))=d$; see Corollary~\ref{cor:k=s}).
It would be natural to ask whether this is also the case for $k=4$ and $s=3$.
\begin{question}\label{question:k=4}
Is it true that for every $d \geq 6$,~ $\od_4(K(d,3)) = 2d-4$?
\end{question}
\noindent
A positive answer to Question~\ref{question:k=4} would imply that for every $t \geq 6$, it is $\NP$-hard to decide whether an input graph $G$ satisfies $\od(G) \leq t$ or $\od(G) \geq 2t-4$, analogously to the hardness result of~\cite{GareyJ76a} for the chromatic number.\footnote{It can be shown, using a result of~\cite{BukhC18}, that every $d \geq 6$ satisfies $\od_4(K(d,3)) \geq \lceil 4d/3 \rceil$ (see Lemma~\ref{lemma:BukhC}). Combining this bound with Theorem~\ref{thm:hard_F}, it follows that for every $t \geq 6$ it is $\NP$-hard to decide whether an input graph $G$ satisfies $\od(G) \leq t$ or $\od(G) \geq \lceil 4t/3 \rceil$.}

We finally consider the algorithmic problem in which given an $n$-vertex $k$-uniform hypergraph $H$ with constant orthogonality dimension, the goal is to find an orthogonal representation of $H$ of as low dimension as possible. It is not difficult to show that a hypergraph $H$ satisfies $\od(H) \leq 2$ if and only if it is $2$-colorable. Hence, by the algorithm of Krivelevich et al.~\cite{KrivelevichNS01}, given an $n$-vertex $3$-uniform hypergraph $H$ with $\od(H) \leq 2$ it is possible to efficiently find a coloring of $H$ that uses $\widetilde{O}(n^{1/5})$ colors, and, in particular, to obtain an orthogonal representation of $H$ of this dimension. For graphs, the first nontrivial case is where we are given as input an $n$-vertex graph $G$ with $\od(G) \leq 3$, for which we prove the following result.

\begin{theorem}\label{thm:alg_od_3}
There exists a randomized polynomial time algorithm that given an $n$-vertex graph $G$ satisfying $\od(G) \leq 3$, finds a coloring of $G$ that uses at most $\widetilde{O}(n^{0.2413})$ colors. In particular, the algorithm finds an orthogonal representation of $G$ of dimension $\widetilde{O}(n^{0.2413})$.
\end{theorem}
\noindent
In fact, we prove a stronger statement than that of Theorem~\ref{thm:alg_od_3}, allowing the input graph $G$ to satisfy $\od_k(G) \leq 3k$ for some integer $k$ (rather than the special case $k=1$; see Theorem~\ref{thm:alg_od_3k}).

\subsection{Overview of Proofs}\label{sec:overview}

\subsubsection{Hardness of Approximating the Orthogonality Dimension of Hypergraphs}

Theorem~\ref{thm:intro-NP-hard-gen} is proved in two steps.
In the first, we show that approximating the orthogonality dimension of $k$-uniform hypergraphs becomes harder as the uniformity parameter $k$ grows, and in the second we prove the hardness result for $4$-uniform hypergraphs. By combining the two, Theorem~\ref{thm:intro-NP-hard-gen} follows. We elaborate below on each of these two steps.

\paragraph{The uniformity reduction.}
Our goal is to show that for every $k_1 \leq k_2$, one can efficiently transform a given $k_1$-uniform hypergraph $H_1$ to a $k_2$-uniform hypergraph $H_2$ so that $\od(H_1)=\od(H_2)$.
We borrow a reduction used in~\cite{GuruswamiHS02} for hypergraph coloring and prove that it preserves the orthogonality dimension.
For simplicity of presentation, let us consider here the case of $k_1=2$ and $k_2=4$. Given an $n$-vertex graph $G=(V,E)$ we construct a $4$-uniform hypergraph $H$ whose vertex set consists of $\ell$ copies $V_1, \ldots, V_\ell$ of $V$.
For $i \in [\ell]$, let $E_i$ denote the collection of $2$-subsets of $V_i$ that correspond to the edges of $G$.
The hyperedges of $H$ are defined as all possible unions of pairs of sets picked from distinct collections $E_i$ and $E_j$.

As a warm-up, we observe that for a sufficiently large $\ell$, say $\ell =n+1$, we have $\chi(G) = \chi(H)$. Indeed, if $G$ is $c$-colorable then the $c$-coloring of $G$ applied to each of the copies of $V$ in $H$ implies that $H$ is $c$-colorable. On the other hand, if $G$ is not $c$-colorable then for every coloring of $H$ by $c$ colors, every graph $(V_i,E_i)$ contains a monochromatic edge. By $\ell > c$, there are $i \neq j$ and sets $e_1 \in E_i$ and $e_2 \in E_j$ such that all vertices of $e_1 \cup e_2$ share the same color. This implies the existence of a monochromatic hyperedge in $H$, hence $H$ is not $c$-colorable.

We next show that for a sufficiently large $\ell$ we have $\od(G) = \od(H)$. The first direction is equally easy, namely, if $G$ has a $t$-dimensional orthogonal representation then by assigning its vectors to every copy of $V$ in $H$ we get a $t$-dimensional orthogonal representation of $H$. For the other direction, assume that $G$ satisfies $\od(G) > t$, and suppose for the sake of contradiction that $\od(H) \leq t$, i.e., there exists a $t$-dimensional orthogonal representation $(u_v)_{v \in V(H)}$ of $H$. By $\od(G) > t$, for every $i \in [\ell]$ there are two vertices $a_i,b_i \in V_i$ adjacent in $G$ whose vectors are not orthogonal, that is, $\langle u_{a_i}, u_{b_i} \rangle \neq 0$. Now, it suffices to show that for some $i \neq j$ the four vectors $u_{a_i}, u_{b_i}, u_{a_j}, u_{b_j}$ are pairwise non-orthogonal, as this would imply a contradiction to the fact that $\{a_i, b_i, a_j, b_j\}$ is a hyperedge of $H$. It is not difficult to see that for some $i \neq j$ the vectors $u_{a_i}$ and $u_{b_j}$ are not orthogonal. Indeed, let $M_1 \in \R^{\ell \times t}$ be the matrix whose rows are the vectors $u_{a_i}$ for $i \in [\ell]$, and let $M_2 \in \R^{\ell \times t}$ be the matrix whose rows are the vectors $u_{b_i}$ for $i \in [\ell]$. Consider the matrix $M \in \R^{\ell \times \ell}$ defined by $M = M_1 \cdot M_2^T$, and notice that its diagonal entries are all nonzero (because $\langle u_{a_i}, u_{b_i} \rangle \neq 0$ for every $i$) and that its rank is at most $t$. Assuming that $\ell > t$, the matrix $M$ must have some nonzero non-diagonal entry (otherwise its rank is $\ell$), implying that $\langle u_{a_i}, u_{b_j} \rangle \neq 0$ for some $i \neq j$. This, however, still does not complete the argument, since it might be the case that for these indices $i$ and $j$, one of the inner products $\langle u_{a_i}, u_{a_j} \rangle$, $\langle u_{b_i}, u_{a_j} \rangle$, and $\langle u_{b_i}, u_{b_j} \rangle$ is zero, avoiding the contradiction.

To overcome this difficulty, we use a symmetrization argument showing that the assumption $\od(H) \leq t$ implies that $H$ has some $t'$-dimensional orthogonal representation $(w_v)_{v \in V(H)}$, where $t'$ is not too large, with the following symmetry property: For every $i$ and $j$, if $\langle w_{a_i}, w_{b_j} \rangle \neq 0$ then the inner products $\langle w_{a_i}, w_{a_j} \rangle$, $\langle w_{b_i}, w_{a_j} \rangle$, and $\langle w_{b_i}, w_{b_j} \rangle$ are all nonzero.
With such an orthogonal representation, applying the above argument with $\ell > t'$ would certainly imply a contradiction and complete the proof. We achieve the symmetry property for $t'=t^4$ using vector tensor products. Namely, we assign to every vertex $a_i$ the vector $w_{a_i} = u_{a_i} \otimes u_{b_i} \otimes u_{a_i} \otimes u_{b_i}$, to every vertex $b_i$ the vector $w_{b_i} = u_{a_i}^{\otimes 2} \otimes u_{b_i}^{\otimes 2}$, and to every other vertex $v$ the vector $w_v = u_v^{\otimes 4}$. It is straightforward to verify that $(w_v)_{v \in V(H)}$ forms a $t'$-dimensional orthogonal representation of $H$. Moreover, by standard properties of tensor products we have
\[\langle w_{a_i}, w_{b_j} \rangle =
\langle u_{a_i}, u_{a_j} \rangle \cdot \langle u_{b_i}, u_{a_j} \rangle \cdot \langle u_{a_i}, u_{b_j} \rangle \cdot \langle u_{b_i}, u_{b_j} \rangle,\]
which can be used to obtain that if $\langle w_{a_i}, w_{b_j} \rangle \neq 0$ then the other three inner products $\langle w_{a_i}, w_{a_j} \rangle$, $\langle w_{b_i}, w_{a_j} \rangle$, and $\langle w_{b_i}, w_{b_j} \rangle$ are nonzero as well.
This completes the proof sketch for $k_1=2$ and $k_2=4$.
For the proof of the general case, we generalize the tensor-based argument to $k$-tuples of vertices (see Lemma~\ref{lemma:TensorORk}) and extend the matrix reasoning applied above using bounds on off-diagonal Ramsey numbers (see Lemma~\ref{lemma:Ramsey} and Remark~\ref{remark:Ramsey}).

\paragraph{Hardness for $4$-uniform hypergraphs.}

We next consider the hardness of approximating the orthogonality dimension of $4$-uniform hypergraphs.
A significant difficulty in proving such a result lies at the challenge of proving strong lower bounds on the orthogonality dimension. For the sake of comparison, in most hardness proofs for hypergraph coloring the lower bound on the chromatic number of the hypergraph $H$ constructed by the reduction is shown by an upper bound on its independence number $\alpha(H)$ and the standard inequality $\chi(H) \geq \frac{|V(H)|}{\alpha(H)}$. This approach cannot be used for the orthogonality dimension, which in certain cases can be exponentially smaller than this ratio (see Proposition~\ref{prop:od_vs_chi}).
Exceptions of this approach, where the lower bound on the chromatic number is not proved via the independence number, are the works of Dinur et al.~\cite{DinurRS05} and Bhangale~\cite{Bhangale18} on which we elaborate next.

A standard technique in proving hardness of approximation results is to reduce from the Label Cover problem, in which given a collection of constraints over a set of variables the goal is to decide whether there exists an assignment that satisfies all the constraints or any assignment satisfies only a small fraction of them (see Section~\ref{sec:LabelCover}). In such reductions, every variable over a domain $[R]$ is encoded via an error-correcting code known as the long code, and the constraints are replaced by ``inner'' constraints designed for the specific studied problem. One way to view the long code is as the graph whose vertices are all subsets of $[R]$ where two sets are adjacent if they are disjoint~\cite{DinurSafra05}. In the hardness proof of~\cite{DinurRS05} for the chromatic number of $3$-uniform hypergraphs, this graph was replaced by the induced subgraph that consists only of subsets of a given size (i.e., a Kneser graph), where a label $\alpha \in [R]$ is encoded by the $2$-coloring of the vertices according to whether the sets contain $\alpha$ or not. The analysis of~\cite{DinurRS05} is crucially based on the large chromatic number of Kneser graphs~\cite{LovaszKneser} and on the property that every coloring of Kneser graphs with number of colors smaller than their chromatic number enforces a large color class that includes a monochromatic edge. The latter property was proved in~\cite{DinurRS05} using the chromatic number of the Schrijver graph~\cite{SchrijverKneser78}, a vertex-critical subgraph of the Kneser graph.
The approach of~\cite{DinurRS05} was recently extended by Bhangale~\cite{Bhangale18}, who used in his long code construction only the vertices of the Schrijver graph. The fact that this subgraph has much fewer vertices and yet large chromatic number has led to improved hardness factors. However, for the analysis to work the ``inner'' constraints had to include four vertices, and this is the reason that the result was obtained for $4$-uniform hypergraphs (and not for $3$-uniform hypergraphs as in~\cite{DinurRS05}).

In the current work we prove the hardness of approximating the orthogonality dimension of $4$-uniform hypergraphs using the reduction applied in~\cite{Bhangale18} for hypergraph coloring. While we achieve the same hardness factors as in~\cite{Bhangale18}, the analysis relies on several different ideas and tools.
This includes the aforementioned symmetrization argument for orthogonal representations (see Lemma~\ref{lemma:TensorORk}), a lower bound of Golovnev et al.~\cite{GolovnevRW17} on the sparsity of low rank matrices with nonzero entries on the diagonal (see Lemma~\ref{lemma:GRW}), and the orthogonality dimension of Schrijver graphs determined in~\cite{Haviv18topo} (see Theorem~\ref{thm:SchrijverGraph}).

\subsubsection{Hardness of Approximating the Orthogonality Dimension of Graphs}

Theorem~\ref{thm:hard_F} relates the hardness of approximating the orthogonality dimension of graphs to orthogonal subspace representations.
Our starting point is the $\NP$-hardness of deciding whether an input graph $G$ satisfies $\od(G) \leq 3$~\cite{Peeters96}. Following an approach of Garey and Johnson~\cite{GareyJ76a}, our reduction constructs a graph $G'$ defined as the lexicographic product of some fixed graph $F$ and the input graph $G$. Namely, we replace every vertex of $F$ by a copy of $G$ and replace every edge of $F$ by a complete bipartite graph between the vertex sets associated with its endpoints (see Definition~\ref{def:lexi}).
We then show that if $\od(G) \leq 3$ then $G'$ has a $\od_3(F)$-dimensional orthogonal representation, whereas if $\od(G) \geq 4$ the orthogonality dimension of $G'$ is at least $\od_4(F)$.
It would be interesting to figure out the best hardness factors that Theorem~\ref{thm:hard_F} can yield (see Question~\ref{question:k=4}).
We note, though, that our approach is limited to multiplicative hardness gaps bounded by $2$, as it is easy to see that every graph $F$ satisfies~ $\od_4(F) \leq \od_3(F)+\od_1(F) \leq 2 \cdot \od_3(F)$.

\subsubsection{Coloring Graphs with Orthogonality Dimension Three}

Consider the problem in which given an $n$-vertex graph $G$ satisfying $\od(G) \leq 3$, the goal is to find an orthogonal representation of $G$ of as low dimension as possible. Employing an approach of~\cite{ChlamtacH14}, we attempt to find a coloring of $G$ with a small number of colors, as this in particular gives an orthogonal representation of the same dimension.
As mentioned before, for every $c \geq 3$ there are known efficient algorithms for coloring $n$-vertex $c$-colorable graphs, however, our only guarantee on $G$ is that its orthogonality dimension is at most $3$.
Interestingly, it follows from a theorem of Kochen and Specker~\cite{KochenS67} (see also~\cite{GodsilZ88}) that the largest possible chromatic number of such a graph is $4$.
It follows that given an $n$-vertex graph $G$ with $\od(G) \leq 3$, one can simply apply the efficient algorithm of~\cite{HalperinNZ02} for coloring $4$-colorable graphs to obtain a coloring of $G$ by $\widetilde{O}(n^{\gamma})$ colors where $\gamma = 7/19 \approx 0.368$.

To improve on this bound, we show an efficient algorithm that finds a large independent set in a given graph $G$ satisfying $\od(G) \leq 3$.
We consider two cases according to the maximum degree in the graph.
If $G$ has a vertex of large degree then the algorithm finds a large independent set in its neighborhood. This can be done since the assumption $\od(G) \leq 3$ implies that the neighborhood of every vertex of $G$ is $2$-colorable (just as for $3$-colorable graphs).
Otherwise, in case that all the degrees in $G$ are small, we use an algorithm of Karger et al.~\cite{KargerMS98} based on a semidefinite relaxation of the chromatic number, called the {\em vector chromatic number}. Our analysis relies on a result by Lov\'asz~\cite{Lovasz79} relating the (strict) vector chromatic number of graphs to their orthogonality dimension.
Now, by repeatedly omitting independent sets in $G$, we obtain a coloring that uses $\widetilde{O}(n^{1/4})$ colors.
This can be slightly improved to $\widetilde{O}(n^{0.2413})$ by applying the refined analysis of Arora et al.~\cite{AroraCC06} for the rounding algorithm of~\cite{KargerMS98}.

As already mentioned, our algorithm can handle any graph $G$ that satisfies $\od_k(G) \leq 3k$ for some integer $k$, rather than for $k=1$ (see Theorem~\ref{thm:alg_od_3k}). The generalized analysis involves a connection, recently proved by Bukh and Cox~\cite{BukhC18}, between the (strict) vector chromatic number and the graph parameters $\od_k$.

\subsection{Outline}

The rest of the paper is organized as follows.
In Section~\ref{sec:preliminaries}, we provide some background on the orthogonality dimension and on the Kneser and Schrijver graphs.
In Section~\ref{sec:uniformity}, we present and analyze the uniformity reduction.
In Section~\ref{sec:hyper_hard}, we prove our hardness result for $4$-uniform hypergraphs and derive Theorem~\ref{thm:intro-NP-hard-gen}. In the final Section~\ref{sec:graph_hard_alg}, we prove our hardness and algorithmic results on the orthogonality dimension of graphs, confirming Theorems~\ref{thm:hard_F} and~\ref{thm:alg_od_3}.

\section{Preliminaries}\label{sec:preliminaries}

\subsection{Orthogonality Dimension}

We define the orthogonality dimension of hypergraphs over a general field.
\begin{definition}
A $t$-dimensional {\em orthogonal representation} of a hypergraph $H=(V,E)$ over a field $\Fset$ is an assignment of a vector $u_v \in \Fset^t$ with $\langle u_v,u_v \rangle \neq 0$ to every vertex $v \in V$, such that for every hyperedge $e \in E$ there are two vertices $v,v' \in e$ satisfying $\langle u_v, u_{v'} \rangle = 0$. The {\em orthogonality dimension} of a hypergraph $H=(V,E)$ over $\Fset$, denoted by $\od(H, \Fset)$, is the smallest integer $t$ for which there exists a $t$-dimensional orthogonal representation of $H$ over $\Fset$. For the real field $\R$, we let $\od(H)$ stand for $\od(H,\R)$.
\end{definition}

The following proposition shows that there are graphs whose orthogonality dimension is exponentially smaller than their chromatic number.

\begin{proposition}\label{prop:od_vs_chi}
There exists an explicit family of graphs $G_t$ such that $\od(G_t) \leq t$ and $\chi(G_t) \geq 2^{\Omega(t)}$.
\end{proposition}

\begin{proof}
Let $t$ be an integer divisible by $4$.
Consider the graph $G_t=(V,E)$ whose vertices are all the $\frac{t}{2}$-subsets of $[t]$ where two sets $A,B \in V$ are adjacent if their intersection size satisfies $|A \cap B| = \frac{t}{4}$. Notice that by $|A|=|B|=\frac{t}{2}$, this condition is equivalent to $|A \bigtriangleup B| = \frac{t}{2}$. Assign to every vertex $A$ the vector $u_A \in \{ \pm 1\}^t$ where $(u_A)_i = +1$ if $i \in A$ and $(u_A)_i = -1$ otherwise. We claim that $(u_A)_{A \in V}$ is an orthogonal representation of $G_t$. Indeed, for every adjacent vertices $A,B \in V$ we have
\[\langle u_A, u_B \rangle = (-1) \cdot |A \bigtriangleup B| + (t - |A \bigtriangleup B|) = t-2 \cdot |A \bigtriangleup B|=0.\]
This implies that $\od(G_t) \leq t$. On the other hand, by a celebrated result of Frankl and R{\"o}dl~\cite{FranklR87}, $\alpha(G_t) \leq 2^{c \cdot t}$ for some $c<1$, implying that
\[\chi(G_t) \geq \frac{|V|}{\alpha(G_t)} \geq \frac{\binom{t}{t/2}}{2^{c \cdot t}} \geq 2^{(1-c-o(1)) \cdot t},\]
completing the proof.
\end{proof}

\subsection{Kneser and Schrijver Graphs}\label{sec:Kneser}

For integers $d \geq 2s$, the {\em Kneser graph} $K(d,s)$ is the graph whose vertices are all the $s$-subsets of $[d]$, where two sets are adjacent if they are disjoint.

A set $A \subseteq [d]$ is said to be {\em stable} if it does not contain two consecutive elements modulo $d$ (that is, if $i \in A$ then $i+1 \notin A$, and if $d \in A$ then $1 \notin A$). In other words, a stable subset of $[d]$ is an independent set in the cycle $C_d$ with the numbering from $1$ to $d$ along the cycle.
For integers $d \geq 2s$, the {\em Schrijver graph} $S(d,s)$ is the graph whose vertices are all the stable $s$-subsets of $[d]$, where two sets are adjacent if they are disjoint.

The orthogonality dimension of the Kneser and Schrijver graphs was determined in~\cite{Haviv18topo} using topological methods.
\begin{theorem}[\cite{Haviv18topo}]\label{thm:SchrijverGraph}
For every $d \geq 2s$,~ $\od(K(d,s)) = \od(S(d,s)) = d-2s+2$.
\end{theorem}

The number of vertices in $K(d,s)$ is clearly $\binom{d}{s}$. We need the following simple bound, given in~\cite{DinurRS05}, on the number of vertices in $S(d,s)$.
\begin{lemma}[\cite{DinurRS05}]\label{lemma:SchrijverV}
For every $d \geq 2s$, the number of vertices in $S(d,s)$ is at most $\binom{d}{d-2s}$.
\end{lemma}

\section{The Uniformity Reduction}\label{sec:uniformity}

In this section we reduce the problem of approximating the orthogonality dimension of $k_1$-uniform hypergraphs over a field $\Fset$ to approximating it on $k_2$-uniform hypergraphs, where $k_1 \leq k_2$.

\subsection{Ramsey Numbers}

For integers $s$ and $t$, the {\em Ramsey number} $R(s,t)$ is defined as the smallest integer $n$ such that every $n$-vertex graph contains a clique of size $s$ or an independent set of size $t$ (or both).
We need the following well-known upper bound on Ramsey numbers due to Erd{\"{o}}s and Szekeres~\cite{ErdosS35}.

\begin{theorem}[\cite{ErdosS35}]\label{thm:Ramsey_ES}
For all integers $s$ and $t$,~ $R(s,t) \leq \binom{t+s-2}{s-1}$.
\end{theorem}

The following lemma relates Ramsey numbers to nonzero patterns of matrices with nonzero entries on the diagonal.

\begin{lemma}\label{lemma:Ramsey}
Let $\Fset$ be a field, and let $M$ be an $n \times n$ matrix over $\Fset$ with nonzero entries on the diagonal.
Denote $t = {\rank}_\Fset(M)$ and suppose that $n \geq R(s,t+1)$.
Then there exists a set $C \subseteq [n]$ of size $|C|=s$ such that for every $i,i' \in C$ it holds that $M_{i,i'} \neq 0$ or $M_{i',i} \neq 0$ (or both).
\end{lemma}

\begin{proof}
Let $M \in \Fset^{n \times n}$ be a matrix with nonzero entries on the diagonal such that $t= {\rank}_\Fset(M)$ and $n \geq R(s,t+1)$.
Define a graph $G$ on the vertex set $[n]$ where two distinct vertices $i$ and $i'$ are adjacent if $M_{i,i'} \neq 0$ or $M_{i',i} \neq 0$ (or both).
Observe that any principal sub-matrix of $M$ that corresponds to an independent set in $G$ is diagonal and has nonzero entries on the diagonal. This implies that the size of any independent set in $G$ does not exceed the rank $t$ of $M$, that is, $\alpha(G) < t+1$. Since $G$ has at least $R(s,t+1)$ vertices, it follows that $G$ contains a clique $C$ of size $s$. This set $C$ satisfies the requirement of the lemma.
\end{proof}

\subsection{Symmetrization Lemma}

We turn to prove our symmetrization lemma for orthogonal representations, which is used in the analysis of the uniformity reduction (Section~\ref{sec:symmetry}) and in our hardness proof for $4$-uniform hypergraphs (Section~\ref{sec:hyper_hard}).
The lemma says, roughly speaking, that given an orthogonal representation of a hypergraph $H$ and a collection $A$ of pairwise disjoint $k$-tuples of vertices of $H$, one can construct another orthogonal representation of $H$ with a polynomially related dimension such that the following property holds: For every two $k$-tuples $a$ and $b$ in $A$ such that the vectors associated with the first vertex of $a$ and the second vertex of $b$ are not orthogonal, it holds that all the vectors associated with the vertices of $a$ are not orthogonal to all the vectors associated with the vertices of $b$. The lemma is used to prove lower bounds on the orthogonality dimension of hypergraphs.
Indeed, in order to show that an assignment of vectors to the vertices does not form an orthogonal representation, one has to show that the vectors associated with the vertices of some hyperedge are pairwise non-orthogonal.
For a hyperedge that consists of vectors from some $k$-tuples of $A$, the lemma essentially allows us to consider only the first two vertices of every such $k$-tuple.

In the proof of the lemma we use the notion of tensor product of vectors.
Recall that for a field $\Fset$, the tensor product of the vectors $u \in \Fset^n$ and $u' \in \Fset^{n'}$, denoted by $u \otimes u'$, is a vector in $\Fset^{n \cdot n'}$ with coordinates corresponding to all products $u_i \cdot u'_{i'}$ for $i \in [n]$ and $i' \in [n']$. We let $u^{\otimes k}$ denote the tensor product of the vector $u$ with itself $k$ times.
It is well known and easy to verify that for every $u,w \in \Fset^n$ and $u',w' \in \Fset^{n'}$ it holds that $\langle u \otimes u', w \otimes w' \rangle = \langle u,w \rangle \cdot \langle u',w' \rangle$.
We also need the following notation.
For a set $V$ and a $k$-tuple $a \in V^k$ we let $a[i]$ stand for the $i$th component of $a$, that is, $a = (a[1],a[2],\ldots,a[k])$.

\begin{lemma}\label{lemma:TensorORk}
Let $k \geq 2$ be an integer, let $H=(V,E)$ be a hypergraph with a $t$-dimensional orthogonal representation $(u_v)_{v \in V}$ over a field $\Fset$, and let $A \subseteq V^k$ be a collection of pairwise disjoint $k$-tuples of vertices of $H$.
Suppose that for every $a \in A$ and $i,i' \in [k]$, $\langle u_{a[i]} , u_{a[i']} \rangle \neq 0$.
Then there exists a $t^{k^2}$-dimensional orthogonal representation $(w_v)_{v \in V}$ of $H$ over $\Fset$ such that
\begin{enumerate}
  \item for every $a \in A$ and $i,i' \in [k]$, $\langle w_{a[i]} , w_{a[i']} \rangle \neq 0$, and
  \item for every (distinct) $a,b \in A$ such that $\langle w_{a[1]} , w_{b[2]} \rangle \neq 0$ it holds that $\langle w_{a[i]} , w_{b[i']} \rangle \neq 0$ for all $i,i' \in [k]$.
\end{enumerate}
\end{lemma}

\begin{proof}
For a hypergraph $H=(V,E)$ and a collection $A \subseteq V^k$ of pairwise disjoint $k$-tuples of vertices of $H$, let $(u_v)_{v \in V}$ be a $t$-dimensional orthogonal representation of $H$ over $\Fset$ as in the lemma.
We assign to the vertices of $H$ $t^{k^2}$-dimensional vectors $(w_v)_{v \in V}$ over $\Fset$ as follows.
For every $a \in A$, define
\begin{eqnarray*}
w_{a[1]} &=& (u_{a[1]} \otimes u_{a[2]} \otimes \cdots \otimes u_{a[k]})^{\otimes k}, \\
w_{a[2]} &=& u_{a[1]}^{\otimes k} \otimes u_{a[2]}^{\otimes k} \otimes \cdots \otimes u_{a[k]}^{\otimes k}.
\end{eqnarray*}
For every other vertex $v \in V$, i.e., a vertex that does not appear in the first two coordinates of the $k$-tuples of $A$, we define $w_v = u_v^{ \otimes k^2}$.
Note that the assignment is well defined since the $k$-tuples of $A$ are pairwise disjoint.
(In fact, we could weaken the condition of pairwise disjointness in the lemma and allow $a[i]=b[i']$ for $a,b \in A$ and $i,i' \geq 3$.)

We first observe that the assignment $(w_v)_{v \in V}$ is an orthogonal representation of $H$ over $\Fset$.
For every vertex $v \in V$, the vector $w_v$ is a tensor product of $k^2$ vectors of the form $u_{v'}$ with $v' \in V$. This implies that $\langle w_v, w_v \rangle$ is a product of inner products of the form $\langle u_{v'}, u_{v'} \rangle$ with $v' \in V$, and since they are all nonzero it follows that $\langle w_v, w_v \rangle \neq 0$.
Now, since $(u_v)_{v \in V}$ is an orthogonal representation of $H$, it suffices to show that for every two vertices $v,v' \in V$, it holds that $\langle w_{v}, w_{v'} \rangle = 0$ whenever $\langle u_{v}, u_{v'} \rangle = 0$.
Consider two vertices $v,v' \in V$ such that $\langle u_{v}, u_{v'} \rangle = 0$.
Notice that $w_v$ and $w_{v'}$ are tensor products of $k^2$ vectors, the $(k+1)$th of which is $u_v$ and $u_{v'}$ respectively, hence $\langle w_{v}, w_{v'} \rangle = 0$.

We next verify that the orthogonal representation $(w_v)_{v \in V}$ of $H$ over $\Fset$ satisfies the properties required by the lemma.
For the first item, observe that for every $a \in A$ each vector $w_{a[i]}$ is a tensor product of $k^2$ vectors of the form $u_{a[j]}$ with $j \in [k]$. Hence, the inner product $\langle w_{a[i]} , w_{a[i']} \rangle$ for $i,i' \in [k]$ is a product of inner products of the form $\langle u_{a[j]} , u_{a[j']} \rangle$ with $j,j' \in [k]$, which are all nonzero by assumption, hence $\langle w_{a[i]} , w_{a[i']} \rangle \neq 0$.
For the second item, consider distinct $a,b \in A$ such that $\langle w_{a[1]} , w_{b[2]} \rangle \neq 0$.
Observe that
\begin{eqnarray*}
\langle w_{a[1]}, w_{b[2]} \rangle = \prod_{(j,j') \in [k] \times [k]}{ \langle u_{a[j]}, u_{b[j']} \rangle},
\end{eqnarray*}
which implies that $\langle u_{a[j]} , u_{b[j']} \rangle \neq 0$ for all $j,j' \in [k]$.
Since every inner product $\langle w_{a[i]} , w_{b[i']} \rangle$ for $i,i' \in [k]$ is a product of inner products of the form $\langle u_{a[j]} , u_{b[j']} \rangle$ with $j,j' \in [k]$, we derive that $\langle w_{a[i]} , w_{b[i']} \rangle \neq 0$, and we are done.
\end{proof}

\subsection{The Reduction}\label{sec:symmetry}

We are ready to prove the main result of this section.

\begin{theorem}\label{thm:reductions}
Let $k_2 \geq k_1 \geq 2$ be constants.
For every parameter $m=m(n) \leq n$ there exists a polynomial time reduction that given a $k_1$-uniform hypergraph $H_1$ on $n$ vertices outputs a $k_2$-uniform hypergraph $H_2$ on $n \cdot m^{O(1)}$ vertices such that for every field $\Fset$,
\begin{enumerate}
  \item\label{itm:1_reduction} $\od(H_2, \Fset) \leq \od(H_1, \Fset)$, and
  \item\label{itm:2_reduction} if $\od(H_2,\Fset) \leq m$ then $\od(H_1, \Fset) \leq \od(H_2, \Fset)$.
\end{enumerate}
\end{theorem}

\begin{proof}
Let $k_2 \geq k_1 \geq 2$ be constants and put $s = \lceil k_2/k_1 \rceil$.
For a given parameter $m=m(n) \leq n$ consider the reduction that given a $k_1$-uniform hypergraph $H_1=(V,E)$ on $n$ vertices outputs the $k_2$-uniform hypergraph $H_2=(V',E')$ defined as follows. Its vertex set is defined as $V' = V_1 \cup \cdots \cup V_\ell$ where each $V_i$ is a copy of the vertex set $V$ of $H_1$ and
\begin{eqnarray}\label{def:ell}
\ell = \binom{m^{k_1^2}+s-1}{s-1}.
\end{eqnarray}
Let $E_i$ denote the collection of $k_1$-subsets of $V_i$ that correspond to the hyperedges in the hypergraph $H_1$.
A hyperedge of $H_2$ is given by the union of $s-1$ sets picked from different $E_i$'s together with additional $k_2-(s-1) \cdot k_1 \leq k_1$ vertices chosen from a set picked from another $E_i$.
More precisely, for every distinct indices $i_1,\ldots,i_s$, choices of $e_{i_j} \in E_{i_j}$ for $j \in [s]$, and a set $e'_{i_s} \subseteq e_{i_s}$ of size $|e'_{i_s}| = k_2-(s-1) \cdot k_1$ we include in $H_2$ the hyperedge $e_{i_1} \cup \cdots \cup e_{i_{s-1}} \cup e'_{i_s}$.
Note that $H_2$ is a $k_2$-uniform hypergraph on $n \cdot \ell = n \cdot O(m^{k_1^2 \cdot (s-1)})$ vertices and that it can be constructed in polynomial running time.

To prove Item~\ref{itm:1_reduction} of the theorem, assume that there exists a $t$-dimensional orthogonal representation of $H_1$ over $\Fset$. For the hypergraph $H_2$, assign the same vectors to each of the $\ell$ copies of the vertex set $V$. Since every hyperedge of $H_2$ contains a $k_1$-subset from some $E_i$, two of the vectors associated with its vertices are orthogonal. It follows that this assignment is a $t$-dimensional orthogonal representation of $H_2$ over $\Fset$, hence $\od(H_2, \Fset) \leq \od(H_1, \Fset)$.

To prove Item~\ref{itm:2_reduction}, denote $t = \od(H_2, \Fset) \leq m$ and let $(u_v)_{v \in V'}$ be a $t$-dimensional orthogonal representation of $H_2$ over $\Fset$. Assume by contradiction that $\od(H_1,\Fset) > t$. Then, the restriction $(u_v)_{v \in V_i}$ of the given orthogonal representation to any $V_i$ does not form an orthogonal representation of the hypergraph $(V_i,E_i)$. This implies that for every $i \in [\ell]$ there exists a $k_1$-tuple $a_i \in V_i^{k_1}$ whose vertices form a hyperedge in $(V_i,E_i)$ such that $\langle u_{a_i[j]}, u_{a_i[j']} \rangle \neq 0$ for all $j,j' \in [k_1]$.

Applying Lemma~\ref{lemma:TensorORk} to the orthogonal representation $(u_v)_{v \in V'}$ of $H_2$ over $\Fset$ with the pairwise disjoint $k_1$-tuples $\{ a_i \mid i \in [\ell] \}$, we obtain a $t^{k_1^2}$-dimensional orthogonal representation $(w_v)_{v \in V'}$ of $H_2$ over $\Fset$ that satisfies the following properties:
\begin{enumerate}
  \item\label{itm:1''} $\langle w_{a_i[j]}, w_{a_i[j']} \rangle \neq 0$ for every $i \in [\ell]$ and $j,j' \in [k_1]$, and
  \item\label{itm:2''} for every (distinct) $i,i' \in [\ell]$ such that $\langle w_{a_i[1]}, w_{a_{i'}[2]} \rangle \neq 0$ it holds that $\langle w_{a_i[j]}, w_{a_{i'}[j']} \rangle \neq 0$ for all $j,j' \in [k_1]$.
\end{enumerate}
Let $M_1 \in \Fset^{\ell \times t^{k_1^2}}$ be the matrix whose rows are the vectors $w_{a_i[1]}$ for $i \in [\ell]$, and let $M_2 \in \Fset^{\ell \times t^{k_1^2}}$ be the matrix whose rows are the vectors $w_{a_i[2]}$ for $i \in [\ell]$. Consider the matrix $M \in \Fset^{\ell \times \ell}$ defined by $M = M_1 \cdot M_2^T$, and notice that $M_{i,i'} = \langle w_{a_i[1]}, w_{a_{i'}[2]} \rangle$ for every $i,i' \in [\ell]$ and that its rank is at most $t^{k_1^2}$. Property~\eqref{itm:1} of the orthogonal representation $(w_v)_{v \in V'}$ implies that the diagonal entries of $M$ are all nonzero.
By the Erd{\"{o}}s-Szekeres bound on Ramsey numbers (Theorem~\ref{thm:Ramsey_ES}) and the definition of $\ell$ in~\eqref{def:ell}, we have $\ell \geq R(s,m^{k_1^2}+1) \geq R(s,t^{k_1^2}+1)$. Hence, we can apply Lemma~\ref{lemma:Ramsey} to obtain that there exists a set $C \subseteq [\ell]$ of size $|C| = s$ such that for every $i,i' \in C$ it holds that $M_{i,i'} \neq 0$ or $M_{i',i} \neq 0$.

To complete the proof recall that the hypergraph $H_2$ includes a hyperedge whose vertices all appear in the $k_1$-tuples $a_i$ with $i \in C$.
We will get a contradiction by proving that the vectors assigned to these vertices by the orthogonal representation $(w_v)_{v \in V'}$ of $H_2$ are pairwise non-orthogonal.
Indeed, by Property~\eqref{itm:1''} of $(w_v)_{v \in V'}$, we have $\langle w_{a_i[j]}, w_{a_i[j']} \rangle \neq 0$ for every $i \in C$ and $j,j' \in [k_1]$.
In addition, for distinct $i,i' \in C$ it holds that $M_{i,i'} \neq 0$ or $M_{i',i} \neq 0$, that is, $\langle w_{a_i[1]}, w_{a_{i'}[2]} \rangle \neq 0$ or $\langle w_{a_{i'}[1]}, w_{a_i[2]} \rangle \neq 0$. By Property~\eqref{itm:2''} of $(w_v)_{v \in V'}$ we have $\langle w_{a_i[j]}, w_{a_{i'}[j']} \rangle \neq 0$ for all $i,i' \in C$ and $j,j' \in [k_1]$, and we are done.
\end{proof}

\begin{remark}\label{remark:Ramsey}
We note that the use of Ramsey numbers in the proof of Theorem~\ref{thm:reductions} is not essential. As pointed out to us by an anonymous reviewer, to prove the assertion of the theorem it suffices to analyze the reduction for the case $k_2=k_1+1$ which can be repeatedly applied to imply the general result. For this special case, the Ramsey numbers are not needed as demonstrated in the discussion in Section~\ref{sec:overview}.
Yet, we have decided to present here the direct analysis for general $k_1$ and $k_2$ since it shows that the reduction used in~\cite{GuruswamiHS02} for hypergraph coloring perfectly preserves the orthogonality dimension, and because the running time of this reduction is slightly better than that of repeatedly applying the reduction for $k_2=k_1+1$.
\end{remark}

We derive the following corollary.

\begin{corollary}\label{cor:reductions}
Let $k_2 \geq k_1 \geq 2$ be constants.
There exists a polynomial time reduction that given a $k_1$-uniform hypergraph $H_1$ on $n$ vertices outputs a $k_2$-uniform hypergraph $H_2$ on $n^{O(1)}$ vertices such that for every field $\Fset$,~ $\od(H_1, \Fset) = \od(H_2, \Fset)$.
\end{corollary}

\begin{proof}
Apply Theorem~\ref{thm:reductions} with $m(n)=n$.
By Item~\ref{itm:1_reduction}, we have $\od(H_2, \Fset) \leq \od(H_1, \Fset) \leq n$. Combining this with Item~\ref{itm:2_reduction}, we get that $\od(H_1, \Fset) \leq \od(H_2, \Fset)$, and we are done.
\end{proof}

\section{The Orthogonality Dimension of Hypergraphs}\label{sec:hyper_hard}

We prove the following hardness result for the orthogonality dimension of $4$-uniform hypergraphs over the real field.

\begin{theorem}\label{thm:NP-hard-gen}
~
\begin{enumerate}
  \item\label{itm:1-hard} There exists a constant $\delta>0$ for which it is $\NP$-hard to decide whether an input $n$-vertex $4$-uniform hypergraph $H$ satisfies $\od(H) \leq 2$ or $\od(H) \geq \log ^\delta n$.
  \item\label{itm:2-hard} Assuming $\NP \nsubseteq \DTIME(n^{O(\log \log n)})$, for every constant $c>0$ there is no polynomial time algorithm that decides whether an input $n$-vertex $4$-uniform hypergraph $H$ satisfies $\od(H) \leq 2$ or $\od(H) \geq c \cdot \frac{\log n}{\log \log n}$.
\end{enumerate}
\end{theorem}

\noindent
Theorem~\ref{thm:intro-NP-hard-gen} follows by combining Theorem~\ref{thm:NP-hard-gen} with Corollary~\ref{cor:reductions}.

\subsection{Label Cover}\label{sec:LabelCover}

Theorem~\ref{thm:NP-hard-gen} is proved by a reduction from the Label Cover problem, defined as follows.

\begin{definition}\label{def:LabelCover}
In the {\em Label Cover} problem the input $\calL = (U,V,E,R,L,\phi)$ consists of a (bi-regular) graph $(U,V,E)$ where every vertex of $U$ is associated with a variable with range $[R]$ and every vertex of $V$ is associated with a variable with range $[L]$. Every edge $(x,z) \in E$ is associated with a projection constraint $\phi_{x \rightarrow z}:[R] \rightarrow [L]$. Given an instance $\calL$ the goal is to find an assignment $\rho$ to the variables of $U \cup V$ that maximizes the number of edges $(x,z) \in E$ such that $\phi_{x \rightarrow z}(\rho(x)) = \rho(z)$.
\end{definition}

The following theorem follows from the PCP theorem~\cite{AroraLMSS98,AroraS98} combined with the parallel repetition theorem~\cite{Raz98}.
\begin{theorem}\label{thm:PCP}
For every $\ell = \ell(n)$ there exists a reduction from $3$-SAT to Label Cover that given a $3$-SAT instance $\varphi$ of size $n$ outputs in running time $n^{O(\ell)}$ a Label Cover instance $\calL$ of size $n^{O(\ell)}$ with range size $2^{O(\ell)}$ such that
\begin{itemize}
  \item if $\varphi$ is satisfiable then there exists an assignment to $\calL$ that satisfies all the constraints, and
  \item if $\varphi$ is not satisfiable then every assignment to $\calL$ satisfies at most $2^{-\Omega(\ell)}$ fraction of the constraints.
\end{itemize}
\end{theorem}

We also need the following result of~\cite{MoshkovitzR10} which gives better parameters for sub-constant soundness error.

\begin{theorem}[\cite{MoshkovitzR10}]\label{thm:PCP_MR}
For every $\eps = \eps(n)$ there exists a reduction from $3$-SAT to Label Cover that given a $3$-SAT instance $\varphi$ of size $n$ outputs in running time $\poly (n, \frac{1}{\eps})$ a Label Cover instance $\calL$ of size $n^{1+o(1)} \cdot \frac{1}{\eps^{O(1)}}$ with range size $\exp( \frac{1}{\eps^{O(1)}})$ such that
\begin{itemize}
  \item if $\varphi$ is satisfiable then there exists an assignment to $\calL$ that satisfies all the constraints, and
  \item if $\varphi$ is not satisfiable then every assignment to $\calL$ satisfies at most $\eps$ fraction of the constraints.
\end{itemize}
\end{theorem}

\subsection{Sparsity of Low Rank Matrices}

The following lemma, proved by Golovnev, Regev, and Weinstein~\cite{GolovnevRW17}, relates the sparsity of a matrix with nonzero entries on the diagonal to its rank.
\begin{lemma}[\cite{GolovnevRW17}]\label{lemma:GRW}
For every field $\Fset$ and an $n \times n$ matrix $M$ over $\Fset$ with nonzero entries on the diagonal, the number $s(M)$ of nonzero entries in $M$ satisfies~ $s(M) \geq \frac{n^2}{4 \cdot \rank_\Fset(M)}$.
\end{lemma}

\subsection{Proof of Theorem~\ref{thm:NP-hard-gen}}

The proof of Theorem~\ref{thm:NP-hard-gen} uses a reduction of Bhangale~\cite{Bhangale18}, described below.

\paragraph{The reduction.}
The reduction gets as input an instance $\calL = (U,V,E,R,L,\phi)$ of the Label Cover problem (see Definition~\ref{def:LabelCover}) and produces a $4$-uniform hypergraph $H$. Let $t$ be an integer parameter to be determined later, and set $s = \lceil (R-t)/2 \rceil$.

The vertices of $H$ include for every variable $x \in U$ a copy $C[x]$ of the vertex set of the Schrijver graph $S(R,s)$ (see Section~\ref{sec:Kneser}). Formally, for every variable $x \in U$ we define $C[x]$ by
\[C[x] = \{x\} \times V(S(R,s)),\]
so that every vertex of $C[x]$ is referred to as a pair $(x,A)$, where $A$ is a stable $s$-subset of $[R]$. The vertex set of $H$ is $V(H) = \cup_{x \in U}{C[x]}$.

We next define the hyperedges of $H$.
We identify here sets with their characteristic vectors, that is, for a set $A$ we let $A_\alpha$ be $1$ if $\alpha \in A$ and $0$ otherwise.
For every two variables $x,y \in U$ with a common neighbor $z \in V$ in $\calL$, the set $\{ (x,A), (x,B), (y,C), (y,D) \}$ forms a hyperedge in $H$ if for every $\alpha,\beta \in [R]$ such that $\phi_{x \rightarrow z}(\alpha) = \phi_{y \rightarrow z}(\beta)$ it holds that $\{A_\alpha, B_\alpha, C_\beta, D_\beta\} = \{0,1\}$.  This completes the description of the reduction.

To provide some intuition, we note that the role of the component $C[x]$ in $H$ is to encode the value of the variable $x$, where an assignment of $\alpha \in [R]$ to $x$ is encoded by the $2$-coloring of $C[x]$ in which every vertex $(x,A)$ is colored by $A_\alpha \in \{0,1\}$. This encoding is used in the following simple proof of the completeness of the reduction.

\begin{proposition}[Completeness]\label{prop:completeness}
If the Label Cover instance $\calL$ is satisfiable then $\od(H) \leq 2$.
\end{proposition}
\begin{proof}
Assume that there exists a satisfying assignment $\rho$ to the variables of $U \cup V$ such that $\phi_{x \rightarrow z}(\rho(x)) = \rho(z)$ for every edge $(x,z)$ of $\calL$. Consider the $2$-coloring of $H$ that assigns to every vertex $(x,A)$ the color $A_{\rho(x)}$. We claim that this coloring is proper. To see this, consider a hyperedge $\{ (x,A), (x,B), (y,C), (y,D) \}$ of $H$ defined with respect to a common neighbor $z \in V$ of $x$ and $y$. Since $\rho$ is a satisfying assignment, we have $\phi_{x \rightarrow z}(\rho(x)) = \rho(z)$ and $\phi_{y \rightarrow z}(\rho(y)) = \rho(z)$, and thus $\phi_{x \rightarrow z}(\rho(x)) = \phi_{y \rightarrow z}(\rho(y))$. By the definition of the hyperedges of $H$ we conclude that $\{ A_{\rho(x)}, B_{\rho(x)}, C_{\rho(y)}, D_{\rho(y)}\} = \{0,1\}$, hence the hyperedge is not monochromatic. This implies that $\chi(H) \leq 2$ and thus $\od(H) \leq 2$, as required.
\end{proof}

We turn to prove the soundness of the reduction.

\begin{proposition}[Soundness]\label{prop:soundness}
If $\od(H) \leq t$ then there exists an assignment to the variables of $U \cup V$ that satisfies at least $\frac{1}{4 \cdot t^6}$ fraction of the edges of $\calL$.
\end{proposition}

\begin{proof}
Assume that $\od(H) \leq t$. Then there exists an assignment of a nonzero vector $u_v \in \R^t$ to every vertex $v \in V(H)$ such that every hyperedge of $H$ contains two vertices assigned to orthogonal vectors. For every variable $x \in U$, the vertices of $C[x]$ can be viewed as the vertices of the Schrijver graph $S(R,s)$, which by Theorem~\ref{thm:SchrijverGraph} satisfies
\[ \od( S(R,s) ) = R-2s+2 = R - 2 \cdot \lceil (R-t)/2 \rceil +2 \geq R - (R-t+1)+2 = t+1.\]
This implies that for every variable $x \in U$ the assignment $(u_v)_{v \in C[x]}$ does not form an orthogonal representation of the graph $S(R,s)$, hence there exist two vertices $a_x = (x,A^{(x)})$ and $b_x = (x,B^{(x)})$ in $C[x]$ such that $A^{(x)} \cap B^{(x)} = \emptyset$ and $\langle u_{a_x}, u_{b_x} \rangle \neq 0$.

Applying Lemma~\ref{lemma:TensorORk} to the orthogonal representation $(u_v)_{v \in V(H)}$ of $H$ with the pairwise disjoint pairs $\{ (a_x,b_x) \mid x \in U \}$, we obtain a $t^4$-dimensional orthogonal representation $(w_v)_{v \in V(H)}$ of $H$ that satisfies the following properties:
\begin{enumerate}
  \item\label{itm:1} for every $x \in U$, $\langle w_{a_x}, w_{b_x} \rangle \neq 0$, and
  \item\label{itm:2} for every $x,y \in U$ such that $\langle w_{a_x}, w_{b_y} \rangle \neq 0$, the inner products $\langle w_{a_x}, w_{a_y} \rangle$, $\langle w_{b_x}, w_{a_y} \rangle$, and $\langle w_{b_x}, w_{b_y} \rangle$ are all nonzero.
\end{enumerate}

We turn to show that there exists an assignment to the variables of $U \cup V$ that satisfies at least $\frac{1}{4 \cdot t^6}$ fraction of the edges of $\calL$.
To this end, we define a {\em randomized assignment} to the variables of $U \cup V$ as follows. For every $x \in U$ consider the set $E(x) = [R] \setminus (A^{(x)} \cup B^{(x)})$. Notice that the disjointness of the sets $A^{(x)}$ and $B^{(x)}$ implies that
\[ |E(x)| = R-2 \cdot s = R - 2\cdot \lceil (R-t)/2 \rceil \leq t.\]
We assign to every variable $x \in U$ an assignment $\rho(x)$ chosen uniformly at random from $E(x)$. We further assign to every variable $z \in V$ an assignment $\rho(z)$ that maximizes the number of constraints involving $z$ that can be satisfied, given that the assignment of every $x \in U$ is chosen from $E(x)$. More formally, for a variable $z \in V$ let $U_z \subseteq U$ be the set of neighbors of $z$ in $U$ and define $\rho(z)$ to be some $\beta \in [L]$ with largest number of variables $x \in U_z$ such that $\beta \in \phi_{x \rightarrow z}(E(x))$.

We need the following three claims.
\begin{claim}\label{claim:ExistsX}
For every variable $z \in V$, there exists a variable $x \in U_z$ for which at least $\frac{1}{4 \cdot t^4}$ fraction of the variables $y$ of $U_z$ satisfies $\langle w_{a_x}, w_{b_y} \rangle \neq 0$.
\end{claim}

\begin{claim}\label{claim:Nonempty}
For every variable $z \in V$ and two variables $x,y \in U_z$ such that $\langle w_{a_x}, w_{b_y} \rangle \neq 0$, we have $\phi_{x \rightarrow z}(E(x)) \cap \phi_{y \rightarrow z}(E(y)) \neq \emptyset$.
\end{claim}

\begin{claim}\label{claim:sets}
Let $\calF$ be a collection of $\ell$-subsets of $[R]$ that includes a set that intersects every set of $\calF$. Then there exists an element of $[R]$ that belongs to at least $\frac{1}{\ell}$ fraction of the sets of $\calF$.
\end{claim}

Let us first show that these three claims complete the proof of the proposition.
We claim that for every variable $z \in V$, the randomized assignment $\rho$ satisfies in expectation at least $\frac{1}{4 \cdot t^6}$ fraction of the constraints involving $z$.
Indeed, fix a variable $z \in V$.
By Claim~\ref{claim:ExistsX}, there exists a variable $x \in U_z$ for which at least $\frac{1}{4 \cdot t^4}$ fraction of the variables $y$ of $U_z$ satisfies $\langle w_{a_x}, w_{b_y} \rangle \neq 0$.
By Claim~\ref{claim:Nonempty}, this $x$ satisfies $\phi_{x \rightarrow z}(E(x)) \cap \phi_{y \rightarrow z}(E(y)) \neq \emptyset$ for at least $\frac{1}{4 \cdot t^4}$ fraction of the variables $y$ of $U_z$.
Applying Claim~\ref{claim:sets} to these sets $\phi_{y \rightarrow z}(E(y))$, we obtain that there exists an element of $[R]$ that belongs to at least $\frac{1}{|E(x)|} \cdot \frac{1}{4 \cdot t^4} \geq \frac{1}{4 \cdot t^5}$ fraction of the sets $\phi_{y \rightarrow z}(E(y))$ with $y \in U_z$.
By the definition of $\rho(z)$, we get that $\rho(z) \in \phi_{y \rightarrow z}(E(y))$ for at least $\frac{1}{4 \cdot t^5}$ fraction of the variables $y \in U_z$. Since $\rho(y)$ is chosen uniformly at random from $E(y)$, $\rho$ satisfies in expectation at least $\frac{1}{t} \cdot \frac{1}{4 \cdot t^5} = \frac{1}{4 \cdot t^6}$ fraction of the constraints involving $z$, as required.
Now, by linearity of expectation, $\rho$ satisfies in expectation at least $\frac{1}{4 \cdot t^6}$ fraction of the constraints of $\calL$, so in particular, there exists an assignment satisfying at least $\frac{1}{4 \cdot t^6}$ fraction of the constraints of $\calL$, and we are done.

It remains to prove Claims~\ref{claim:ExistsX},~\ref{claim:Nonempty}, and~\ref{claim:sets}.

\begin{proof}[ of Claim~\ref{claim:ExistsX}]
Fix a variable $z \in V$ and denote $n = |U_z|$.
Let $M_1 \in \R^{n \times t^4}$ be the matrix whose rows are the vectors $w_{a_x}$ for $x \in U_z$, and let $M_2 \in \R^{n \times t^4}$ be the matrix whose rows are the vectors $w_{b_x}$ for $x \in U_z$. Consider the matrix $M \in \R^{n \times n}$ defined by $M = M_1 \cdot M_2^T$, and notice that $M_{x,y} = \langle w_{a_x}, w_{b_y} \rangle$ for every $x,y \in U_z$ and that its rank is at most $t^4$. Property~\eqref{itm:1} of the orthogonal representation $(w_v)_{v \in V(H)}$ implies that the diagonal entries of $M$ are all nonzero, so we can apply Lemma~\ref{lemma:GRW} to obtain that $s(M) \geq \frac{n^2}{4 \cdot t^4}$. In particular, there exists a row of $M$ with at least $\frac{n}{4 \cdot t^4}$ nonzero entries. This implies that there exists a variable $x \in U_z$ for which at least $\frac{1}{4 \cdot t^4}$ fraction of the variables $y$ of $U_z$ satisfies $\langle w_{a_x}, w_{b_y} \rangle \neq 0$, as required.
\end{proof}

\begin{proof}[ of Claim~\ref{claim:Nonempty}]
Fix a variable $z \in V$ and two variables $x,y \in U_z$ such that $\langle w_{a_x}, w_{b_y} \rangle \neq 0$.
Assume by contradiction that $\phi_{x \rightarrow z}(E(x)) \cap \phi_{y \rightarrow z}(E(y)) = \emptyset$.

We first show that the set $e=\{ (x, A^{(x)}), (x, B^{(x)}), (y, A^{(y)}), (y, B^{(y)})\}$ is a hyperedge of $H$.
To see this, take $\alpha, \beta \in [R]$ such that $\phi_{x \rightarrow z}(\alpha) = \phi_{y \rightarrow z}(\beta)$.
By the assumption that $\phi_{x \rightarrow z}(E(x))$ and $\phi_{y \rightarrow z}(E(y))$ are disjoint, it follows that $\alpha \notin E(x)$ or $\beta \notin E(y)$, that is, $\alpha \in A^{(x)} \cup B^{(x)}$ or $\beta \in A^{(y)} \cup B^{(y)}$. In the former case, using the fact that $A^{(x)} \cap B^{(x)} = \emptyset$, we have $\{ A^{(x)}_\alpha , B^{(x)}_\alpha \} = \{0,1\}$ implying that $e$ is a hyperedge. The latter case is handled similarly.

We now show that the vectors in $\{ w_{a_x}, w_{b_x}, w_{a_y}, w_{b_y}\}$ are pairwise non-orthogonal.
By Property~\eqref{itm:1} of the vectors in $(w_v)_{v \in V(H)}$ we have $\langle w_{a_x}, w_{b_x} \rangle \neq 0$ and $\langle w_{a_y}, w_{b_y} \rangle \neq 0$.
By assumption we have $\langle w_{a_x}, w_{b_y} \rangle \neq 0$, and using Property~\eqref{itm:2} we get that $\langle w_{a_x}, w_{a_y} \rangle$, $\langle w_{b_x}, w_{a_y} \rangle$, and $\langle w_{b_x}, w_{b_y} \rangle$ are all nonzero as well.
Since $e$ is a hyperedge of $H$, this yields a contradiction to the fact that $(w_v)_{v \in V(H)}$ is an orthogonal representation of $H$, and we are done.
\end{proof}

\begin{proof}[ of Claim~\ref{claim:sets}]
Assume that the set $A \in \calF$ intersects every set of $\calF$. Denote $A = \{a_1,\ldots,a_\ell\}$ and put $\calA_i = \{ B \in \calF \mid a_i \in B\}$ for every $i \in [\ell]$. Since $\calF = \cup_{i \in [\ell]}{\calA_i}$, it follows that
\[|\calF| = |\cup_{i \in [\ell]}{\calA_i}| \leq \sum_{i \in [\ell]}{|\calA_i|},\]
implying that there exists an $i \in [\ell]$ such that $|\calA_i| \geq \frac{|\calF|}{\ell}$. Therefore, some $a_i$ belongs to at least $\frac{1}{\ell}$ fraction of the sets of $\calF$, as required.
\end{proof}

The proof of the proposition is completed.
\end{proof}

To derive Theorem~\ref{thm:NP-hard-gen} we just have to set the parameters appropriately, as is done below.

\begin{proof}[ of Theorem~\ref{thm:NP-hard-gen}, Item~\ref{itm:1-hard}]
Let $\eps$ and $t$ be two parameters.
By Theorem~\ref{thm:PCP_MR}, there exists a reduction that maps a 3-SAT instance $\varphi$ of size $n$ in running time $\poly(n,\frac{1}{\eps})$ to a Label Cover instance $\calL$ of size $n^{1+o(1)} \cdot \frac{1}{\eps^{O(1)}}$ with range size $R \leq \exp(\frac{1}{\eps^{O(1)}})$ such that if $\varphi$ is satisfiable then so is $\calL$ and if $\varphi$ is not satisfiable then every assignment to $\calL$ satisfies at most $\eps$ fraction of the constraints. We proceed by mapping the instance $\calL$ to a $4$-uniform hypergraph $H$ using the reduction described above. By Lemma~\ref{lemma:SchrijverV} and the fact that $R-2s \leq t$, the number of vertices in $S(R,s)$ is at most $\binom{R}{R-2s} \leq R^t$. Hence, the number $N$ of vertices in $H$ satisfies
\[N \leq n^{1+o(1)} \cdot \frac{1}{\eps^{O(1)}} \cdot R^t \leq n^{1+o(1)} \cdot \exp \Big (\frac{1}{\eps^{O(1)}} \Big )^t.\]
By Proposition~\ref{prop:completeness}, if $\varphi$ is satisfiable then $\od(H) \leq 2$.
By Proposition~\ref{prop:soundness}, if $\od(H) \leq t$ then there exists an assignment to the variables of $U \cup V$ that satisfies at least $\frac{1}{4 \cdot t^6}$ fraction of the constraints of $\calL$. In particular, for $\eps < \frac{1}{4 \cdot t^6}$, if $\varphi$ is not satisfiable then $\od(H) > t$.

Now, for a sufficiently small constant $\delta>0$ set $t = \log^\delta n$ and, say, $\eps = \frac{1}{\log^{7\delta}n}$, so that $\eps < \frac{1}{4 \cdot t^6}$. For these parameters the running time of the reduction is polynomial in $n$ and the hypergraph $H$ has $N = n^{1+o(1)}$ vertices.
The reduction implies that for some $\delta'>0$ it is $\NP$-hard to decide whether an input $N$-vertex $4$-uniform hypergraph $H$ satisfies $\od(H) \leq 2$ or $\od(H) \geq \log^{\delta'} N$, completing the proof.
\end{proof}

\begin{proof}[ of Theorem~\ref{thm:NP-hard-gen}, Item~\ref{itm:2-hard}]
Let $\ell$ and $t$ be two parameters.
By Theorem~\ref{thm:PCP}, there exists a reduction that maps a 3-SAT instance $\varphi$ of size $n$ in running time $n^{O(\ell)}$ to a Label Cover instance $\calL$ of size $n^{O(\ell)}$ with range size $R \leq 2^{O(\ell)}$ such that if $\varphi$ is satisfiable then so is $\calL$ and if $\varphi$ is not satisfiable then every assignment to $\calL$ satisfies at most $2^{-\Omega(\ell)}$ fraction of the constraints. We proceed by mapping the instance $\calL$ to a $4$-uniform hypergraph $H$ using the reduction described above. Using Lemma~\ref{lemma:SchrijverV}, the number $N$ of vertices in $H$ satisfies
\[N \leq n^{O(\ell)} \cdot R^t \leq  n^{O(\ell)} \cdot 2^{O(\ell \cdot t)}.\]
By Proposition~\ref{prop:completeness}, if $\varphi$ is satisfiable then $\od(H) \leq 2$.
By Proposition~\ref{prop:soundness}, if $\od(H) \leq t$ then there exists an assignment to the variables of $U \cup V$ that satisfies at least $\frac{1}{4 \cdot t^6}$ fraction of the constraints of $\calL$. In particular, for $\ell = \Omega(\log t)$, if $\varphi$ is not satisfiable then $\od(H) > t$.

Now, for an arbitrarily large constant $c'>0$ put $t = c' \cdot \log n$ and define $\ell = c'' \cdot \log \log n$ for a sufficiently large $c''>0$. For these parameters the running time of the reduction is $n^{O(\ell)} = n^{O(\log \log n)}$ and the hypergraph $H$ has $N = n^{O(\log \log n)}$ vertices. The reduction implies that assuming $\NP \nsubseteq \DTIME(n^{O(\log \log n)})$, there is no polynomial time algorithm that decides whether an input $N$-vertex $4$-uniform hypergraph $H$ satisfies $\od(H) \leq 2$ or $\od(H) \geq c \cdot \frac{\log N}{\log \log N}$, where $c>0$ can be arbitrarily large. This completes the proof.
\end{proof}

\section{The Orthogonality Dimension of Graphs}\label{sec:graph_hard_alg}

In this section we focus on the orthogonality dimension of graphs and prove Theorems~\ref{thm:hard_F} and~\ref{thm:alg_od_3}.

\subsection{Orthogonal Subspace Representations}

We generalize the notion of orthogonal representations over the real field by assigning to every vertex a subspace of a given dimension, so that adjacent vertices are assigned to orthogonal subspaces.

\begin{definition}\label{def:ortho_k-subspace}
A $t$-dimensional {\em orthogonal $k$-subspace representation} of a graph $G=(V,E)$ is an assignment of a subspace $U_v \subseteq \R^t$ with $\dim (U_v)=k$ to every vertex $v \in V$, such that the subspaces $U_v$ and $U_{v'}$ are orthogonal whenever $v$ and $v'$ are adjacent in $G$. For a graph $G$, let $\od_k(G)$ denote the smallest integer $t$ for which there exists a $t$-dimensional orthogonal $k$-subspace representation of $G$.
\end{definition}
\noindent
Clearly, $\od(G) = \od_1(G)$ for every graph $G$.
It is also easy to see that the multichromatic numbers of graphs, defined in Section~\ref{sec:coloring}, bound the parameters $\od_k$ from above, namely, $\od_k(G) \leq \chi_k(G)$ for every $G$ and $k$.

\subsection{Hardness}

In this section we prove Theorem~\ref{thm:hard_F}, namely that for every graph $F$, it is $\NP$-hard to decide whether an input graph $G$ satisfies $\od(G) \leq \od_3(F)$ or $\od(G) \geq \od_4(F)$. The proof employs the notion of lexicographic product of graphs, defined as follows.
\begin{definition}\label{def:lexi}
The {\em lexicographic product} of the graphs $G_1 = (V_1,E_1)$ and $G_2 = (V_2,E_2)$, denoted by $G_1 \bullet G_2$, is the graph whose vertex set is $V_1 \times V_2$ where two vertices $(x_1,y_1)$ and $(x_2,y_2)$ are adjacent if either $\{x_1,x_2\} \in E_1$ or $x_1=x_2$ and $\{y_1,y_2\} \in E_2$.
\end{definition}
\noindent
One can view the graph $G_1 \bullet G_2$ as the graph obtained from $G_1$ by replacing every vertex by a copy of $G_2$ and replacing every edge by a complete bipartite graph between the vertex sets associated with its endpoints.
We need the following property of the orthogonality dimension of lexicographic products of graphs.

\begin{lemma}\label{lemma:lexi}
For every two graphs $G_1$ and $G_2$,~ $\od(G_1 \bullet G_2) = \od_k (G_1)$ where $k = \od(G_2)$.
\end{lemma}

\begin{proof}
Let $G_1 = (V_1,E_1)$ and $G_2 = (V_2,E_2)$ be two graphs and denote $k = \od(G_2)$.

We first prove that $\od(G_1 \bullet G_2) \geq \od_k (G_1)$.
Denote $t = \od(G_1 \bullet G_2)$, then there exists a $t$-dimensional orthogonal representation $(u_{(x,y)})_{(x,y) \in V_1 \times V_2}$ of $G_1 \bullet G_2$. For every $x \in V_1$, let $U_x$ denote the subspace of $\R^t$ spanned by all vectors $u_{(x,y)}$ with $y \in V_2$. By the definition of $G_1 \bullet G_2$, the subspaces $U_{x}$ and $U_{x'}$ are orthogonal whenever $x$ and $x'$ are adjacent in $G_1$. Further, for every $x \in V_1$, the restriction of the given orthogonal representation to the copy of $G_2$ associated with $x$ forms an orthogonal representation of $G_2$, so by $k = \od(G_2)$ it follows that $\dim(U_x) \geq k$. This implies that there exists a $t$-dimensional orthogonal $k$-subspace representation of $G_1$, hence $\od_k(G_1) \leq t$, as required.

We next prove that $\od(G_1 \bullet G_2) \leq \od_k (G_1)$.
Denote $t = \od_k(G_1)$, then there exists a $t$-dimensional orthogonal $k$-subspace representation $(U_x)_{x \in V_1}$ of $G_1$.
By $k = \od(G_2)$, there exists a $k$-dimensional orthogonal representation $(u_y)_{y \in V_2}$ of $G_2$.
For every $x \in V_1$, the fact that $\dim(U_x)=k$ implies that there exists an orthogonal linear transformation $T_x$ from $\R^k$ onto $U_x$.
We assign to every vertex $(x,y) \in V_1 \times V_2$ of $G_1 \bullet G_2$ the nonzero vector $w_{(x,y)} = T_x(u_y) \in \R^t$. We claim that this is a $t$-dimensional orthogonal representation of $G_1 \bullet G_2$. To see this, let $(x_1,y_1)$ and $(x_2,y_2)$ be two adjacent vertices in $G_1 \bullet G_2$. If $x_1$ and $x_2$ are adjacent in $G_1$ then the subspaces $U_{x_1}$ and $U_{x_2}$ are orthogonal, hence the vectors $w_{(x_1,y_1)} \in U_{x_1}$ and $w_{(x_2,y_2)} \in U_{x_2}$ are orthogonal as well. Otherwise, $x_1=x_2$ and the vertices $y_1$ and $y_2$ are adjacent in $G_2$. This implies that the vectors $u_{y_1}$ and $u_{y_2}$ are orthogonal, and since $T_{x_1}$ preserves inner products it follows that $T_{x_1}(u_{y_1})$ and $T_{x_1}(u_{y_2})$ are orthogonal, and we are done.
\end{proof}

Equipped with Lemma~\ref{lemma:lexi}, we are ready to prove Theorem~\ref{thm:hard_F}.

\begin{proof}[ of Theorem~\ref{thm:hard_F}]
Fix a graph $F$. We reduce from the $\NP$-hard problem of deciding whether an input graph $G$ satisfies $\od(G) \leq 3$~\cite{Peeters96}. The reduction maps an input graph $G$ to the lexicographic product $G' = F \bullet G$ of $F$ and $G$.
The graph $G'$ can clearly be constructed in polynomial time.
The correctness of the reduction follows from Lemma~\ref{lemma:lexi}.
Indeed, we have $\od(G') = \od_k(F)$ for $k = \od(G)$, so if $\od(G) \leq 3$ then $\od(G') \leq \od_3(F)$ and if $\od(G) \geq 4$ then $\od(G') \geq \od_4(F)$.
\end{proof}

\subsection{Algorithm}

Before presenting our algorithm, we provide some background on the vector chromatic number of graphs.

\subsubsection{Vector Chromatic Number}

Consider the following two relaxations of the chromatic number of graphs, due to Karger, Motwani, and Sudan~\cite{KargerMS98}.

\begin{definition}\label{def:chi_v}
For a graph $G=(V,E)$ the {\em vector chromatic number} of $G$, denoted by $\vchrom(G)$, is the minimal real value of $\kappa > 1$ for which there exists an assignment of a unit real vector $u_v$ to every vertex $v \in V$ such that $\langle u_v, u_{v'} \rangle \leq -\frac{1}{\kappa -1}$ whenever $v$ and $v'$ are adjacent in $G$.
\end{definition}

\begin{definition}
For a graph $G=(V,E)$ the {\em strict vector chromatic number} of $G$, denoted by $\svchrom(G)$, is the minimal real value of $\kappa > 1$
for which there exists an assignment of a unit real vector $u_v$ to every vertex $v \in V$ such that $\langle u_v, u_{v'} \rangle = -\frac{1}{\kappa -1}$ whenever $v$ and $v'$ are adjacent in $G$.
\end{definition}
\noindent
It is well known and easy to verify that for every graph $G$, $\vchrom(G) \leq \svchrom(G) \leq \chi(G)$.
Karger et al.~\cite{KargerMS98} have obtained the following algorithmic result.

\begin{theorem}[\cite{KargerMS98}]\label{thm:KMS}
There exists a randomized polynomial time algorithm that given an $n$-vertex graph $G$ with maximum degree at most $\Delta$ and $\vchrom(G)\leq \kappa$ for some $\kappa \geq 2$, finds an independent set of size $\widetilde{\Omega}(\frac{n}{\Delta^{1-2/\kappa}})$.
\end{theorem}

Note that the well-known Lov\'{a}sz $\vartheta$-function, introduced in~\cite{Lovasz79}, is known to satisfy $\vartheta(G)=\svchrom(\overline{G})$ for every graph $G$~\cite{KargerMS98}.
Combining this with a result of~\cite{Lovasz79}, it follows that the orthogonality dimension forms an upper bound on the strict vector chromatic number, that is, $\od(G) \geq \svchrom(G)$ for every graph $G$.
This was recently generalized by Bukh and Cox as follows (see~\cite[Proposition~23]{BukhC18}).

\begin{lemma}[\cite{BukhC18}]\label{lemma:BukhC}
For every graph $G$ and an integer $k$,~ $\od_k(G) \geq k \cdot \svchrom(G)$.
\end{lemma}

We derive that the graph parameters $\od_k$ and $\chi_k$ coincide on Kneser graphs $K(d,s)$ whenever $k$ is divisible by $s$.

\begin{corollary}\label{cor:k=s}
For all integers $\ell \geq 1$ and $d \geq 2s$,~ $\od_{\ell \cdot s}(K(d,s)) = \ell \cdot d$.
\end{corollary}

\begin{proof}
For the upper bound on $\od_{\ell \cdot s}(K(d,s))$, recall that Conjecture~\ref{conj:Stahl} was confirmed for $k = \ell \cdot s$ in~\cite{Stahl76}, hence $\od_{\ell \cdot s}(K(d,s)) \leq \chi_{\ell \cdot s}(K(d,s)) = \ell \cdot d$.
For the lower bound, combine Lemma~\ref{lemma:BukhC} with the fact that $\svchrom(K(d,s)) = \frac{d}{s}$ (see~\cite{Lovasz79}), to get that $\od_{\ell \cdot s}(K(d,s)) \geq \ell \cdot d$.
\end{proof}

\subsubsection{The Algorithm}

We present an efficient algorithm that given a graph $G$ that satisfies $\od(G) \leq 3$ (or, more generally, $\od_k(G) \leq 3k$ for some $k$), finds a coloring of $G$ with relatively few colors.
We use the following simple claim of Blum~\cite{Blum94} which reduces the algorithmic task of coloring a graph to the algorithmic task of finding a large independent set in it.

\begin{claim}[\cite{Blum94}]\label{claim:Blum}
Let $\mathcal{G}$ be a graph family which is closed under taking induced subgraphs, let $c_1,c_2>1$ be arbitrary fixed constants, and let $f:\N \rightarrow \N$ be any non-decreasing function satisfying $c_1 \cdot f(n)\leq f(2n)\leq c_2 \cdot f(n)$ for all sufficiently large $n$. Then if there exists a (randomized) polynomial time algorithm which finds an independent set of size $f(n)$ in any $n$-vertex graph
$G\in\mathcal{G}$, then there exists a (randomized) polynomial time algorithm which finds an $O(\frac{n}{f(n)})$-coloring of any $n$-vertex graph $G\in\mathcal{G}$.
\end{claim}

We need the following two simple lemmas.

\begin{lemma}\label{lemma:od_2k_2-col}
Let $G=(V,E)$ be a graph such that $\od_k(G) \leq 2k$ for some integer $k$.
Then $G$ is $2$-colorable.
\end{lemma}

\begin{proof}
Let $G=(V,E)$ be a graph satisfying $\od_k(G) \leq 2k$, and let $(U_v)_{v \in V}$ be a $2k$-dimensional orthogonal $k$-subspace representation of $G$.
It suffices to prove that every connected component of $G$ is $2$-colorable. Fix a vertex $v$ in some connected component of $G$, and observe that there exists a unique subspace of $\R^{2k}$ of dimension $k$ orthogonal to $U_{v}$. This implies that the orthogonal subspace representation of $G$ provides a $2$-coloring of the connected component of $v$, where the vertices of even distance from $v$ are assigned to $U_{v}$ and the vertices of odd distance from $v$ are assigned to its orthogonal complement ${U}_{v}^{\perp}$, so we are done.
\end{proof}

For a graph $G$ and a vertex $v$, let $N(v)$ denote the neighborhood of $v$ in $G$ and let $G[N(v)]$ denote the subgraph of $G$ induced by $N(v)$.

\begin{lemma}\label{lemma:od_k_N(v)}
Let $G=(V,E)$ be a graph such that $\od_k(G) \leq 3k$ for some integer $k$.
Then for every vertex $v \in V$ the subgraph $G[N(v)]$ is $2$-colorable.
\end{lemma}

\begin{proof}
Let $G=(V,E)$ be a graph satisfying $\od_k(G) \leq 3k$, and let $(U_v)_{v \in V}$ be a $3k$-dimensional orthogonal $k$-subspace representation of $G$.
Let $v \in V$ be a vertex in $G$. The subspace $U_v$ is orthogonal to all subspaces $U_{v'}$ with $v' \in N(v)$. By $\dim (U_v)=k$, the orthogonal complement of $U_v$ in $\R^{3k}$ has dimension $2k$. It follows that $G[N(v)]$ admits a $2k$-dimensional orthogonal $k$-subspace representation, hence by Lemma~\ref{lemma:od_2k_2-col} it is $2$-colorable, as required.
\end{proof}

We are ready to prove the following result.

\begin{theorem}\label{thm:col_3k_k}
There exists a randomized polynomial time algorithm that given an $n$-vertex graph $G$ that satisfies $\od_k(G) \leq 3k$ for some $k$, finds a coloring of $G$ that uses at most $\widetilde{O}(n^{1/4})$ colors. In particular, the algorithm finds an orthogonal representation of $G$ of dimension $\widetilde{O}(n^{1/4})$.
\end{theorem}

\begin{proof}
By Claim~\ref{claim:Blum} it suffices to show that there exists a randomized polynomial time algorithm that given an $n$-vertex graph $G=(V,E)$ with $\od_k(G) \leq 3k$ for some $k$, finds in $G$ an independent set of size $\widetilde{\Omega}(n^{3/4})$.
We consider two possible cases. Suppose first that there exists a vertex $v \in V$ in $G$ whose degree is at least $\Delta = n^{3/4}$. Then by Lemma~\ref{lemma:od_k_N(v)} the subgraph $G[N(v)]$ is $2$-colorable, so we can find an independent set in $G$ of size at least $\frac{\Delta}{2}$ by finding in polynomial time a $2$-coloring of $G[N(v)]$ and taking a largest color class. Otherwise, the maximum degree of $G$ is at most $\Delta$. By Lemma~\ref{lemma:BukhC} we have
\[\vchrom(G) \leq \svchrom(G) \leq \frac{\od_k(G)}{k} \leq 3,\]
so by Theorem~\ref{thm:KMS} we can find in polynomial time an independent set of size $\widetilde{\Omega}(\frac{n}{\Delta^{1/3}}) \geq \widetilde{\Omega}(n^{3/4})$, and we are done.
\end{proof}

To improve the number of used colors, we employ the following result that stems from the analysis by Arora et al.~\cite{AroraCC06} of the semidefinite relaxation of~\cite{KargerMS98}. (For an explicit statement, see~\cite[Lemma~4.12]{ChlamtacH14}, applied with $\sigma=0.5$, $c \approx 0.04843726$, and $\delta \approx 0.7587$.)
\begin{theorem}[\cite{AroraCC06}]\label{thm:Chl}
There exists a randomized polynomial time algorithm that given an $n$-vertex graph $G$ with maximum degree at most $\Delta = n^{0.7587}$ and $\vchrom(G) \leq 3$, finds an independent set in $G$ of size at least $\widetilde{\Omega}(n \cdot \Delta^{-0.3179}) \geq \widetilde{\Omega}(n^{0.7587})$.
\end{theorem}

By applying Theorem~\ref{thm:Chl} instead of Theorem~\ref{thm:KMS} in the proof of Theorem~\ref{thm:col_3k_k}, we obtain the following slight improvement, confirming Theorem~\ref{thm:alg_od_3}.

\begin{theorem}\label{thm:alg_od_3k}
There exists a randomized polynomial time algorithm that given an $n$-vertex graph $G$ that satisfies $\od_k(G) \leq 3k$ for some $k$, finds a coloring of $G$ that uses at most $\widetilde{O}(n^{0.2413})$ colors.
\end{theorem}

\section*{Acknowledgements}
We are grateful to Alexander Golovnev for useful discussions and to the anonymous reviewers for their valuable suggestions.

\bibliographystyle{abbrv}
\bibliography{xi_hyper}

\end{document}